\documentclass[a4paper, 11pt]{article}
\usepackage[left=2cm,top=1.9cm, bottom=1in, right=2cm]{geometry} 

\usepackage{graphicx,subfigure}
\usepackage{amsfonts,amsmath,amssymb, mathrsfs, amsthm}
\usepackage{multirow}
\usepackage{epic,eepic,eepicemu}
\usepackage{epsf}
\usepackage{epsfig}
\usepackage{graphics}
\usepackage{mathrsfs}
\usepackage{psfrag}
\usepackage{tikz}
\usepackage{url}

\usetikzlibrary{arrows}
\usepackage{verbatim}
\newtheorem{theorem}{Theorem}
\newtheorem{definition}{Definition}

\newtheorem{corollary}{Corollary}
\newtheorem{remark}{Remark}
\newtheorem{example}{Example}
\makeatletter
\newcommand*{\rom}[1]{\expandafter\@slowromancap\romannumeral #1@}
\makeatother

\DeclareMathOperator*{\Motimes}{\text{\raisebox{0.25ex}{\scalebox{0.8}{$\bigotimes$}}}}
\usepackage[skip=2pt,font=footnotesize]{caption}

\begin{document}

\title{Zero Error Coordination}

\author{Mahed~Abroshan$^*$, Amin~Gohari$^*$ and Sidharth~Jaggi$^\dagger$\\$*$ ISSL Lab, Department of Electrical  Engineering, Sharif University of Technology,
\\ $\dagger$ Department of Information Engineering, The Chinese University of Hong Kong.}

\allowdisplaybreaks
\maketitle

\begin{abstract}
In this paper, we consider a zero error coordination problem wherein the nodes of a network exchange messages to be able to perfectly coordinate their actions with the individual observations of each other. While previous works on coordination commonly assume an asymptotically vanishing error, we assume \emph{ exact, zero error} coordination. Furthermore, unlike previous works that employ the empirical or strong notions of coordination, we define and use a notion of \emph{set coordination.}  This notion of coordination bears similarities with the empirical notion of coordination. We observe that set coordination, in its special case of  two nodes with a one-way communication link is equivalent with the ``Hide and Seek" source coding problem of McEliece and Posner. The Hide and Seek problem has known intimate connections with graph entropy, rate distortion theory,  R\'enyi mutual information and even error exponents. Other special cases of the set coordination problem relate to Witsenhausen's zero error rate and the  distributed computation problem. These connections motivate a better understanding of set coordination, its connections with empirical coordination, and its study in more general setups. This paper takes a first step in this direction by proving new results for two node networks. 
\end{abstract}
%\IEEEpeerreviewmaketitle

\section{Introduction}

Consider a network where each node of the network has a private observation and needs to produce an action. These actions should be coordinated with the observations; therefore some form of communication is necessary among the nodes. The fundamental limits of the required communication was originally studied by Cuff et al.~in  \cite{cuff} where the authors assumed that the observation of node $i$ ($1\leq i\leq m$) are  i.i.d.~repetitions of some random variable $X_i$. The joint distribution of $(X_1, X_2, ..., X_m)$ was assumed to be a given. Denoting the action of the $i$-th node by $Y_i$, coordination was then modeled as requiring the joint pmf of the outputs conditioned on the inputs to be very close to some given $p(y_1, y_2, ..., y_m|x_1, x_2, ...,x_m)$. Here the authors introduce two notions of empirical and strong coordination: in the strong coordination, memoryless repetitions of the channel $p(y_1, y_2, ..., y_m|x_1, x_2, ...,x_m)$ are simulated, whereas in empirical coordination, only the data histograms (or its joint type) is equal  $p(y_1, y_2, ..., y_m, x_1, x_2, ...,x_m)$. 
The common theme is that the conditional pmf $p(y_1, y_2, ..., y_m|x_1, x_2, ...,x_m)$ is approximated \emph{asymptotically} as the number of i.i.d.~observations (the block length) goes to infinity. In this work, however, we are interested in exact zero error coordination, \emph{i.e.,} coordination should be achieved with probability one. In this way, our work is related to \cite[Sec IV]{exact}, \cite{Winter} on exact strong coordination capacity, however we adopt a different \emph{set coordination} criterion (which is closer to the empirical notion of coordination). 

In our setup, for any observation vector $(x_1, x_2, ..., x_m)$ by the $m$ nodes of a network, we assume a permissible set of output actions $(y_1, y_2, ..., y_m)$. In other words, we are not directly interested in simulating a given $$p(y_1, y_2, ..., y_m|x_1, x_2, ...,x_m).$$ Rather, for every $(x_1, x_2, ..., x_m)$, we define a set $$\mathcal{A}_{x_1, x_2, ..., x_m}\subseteq \mathcal{Y}_1\times \mathcal{Y}_2\times\cdots\times \mathcal{Y}_m,$$
such that $(y_1, y_2, ..., y_m)\in\mathcal{A}_{x_1, x_2, ..., x_m}$. Here $\mathcal{Y}_i$ is the action set of node $i$. We call this a ``set coordination." In Section \ref{sec:set-coordination}, we compare set coordination with empirical coordination.
\begin{example} If $|\mathcal{A}_{x_1, x_2, ..., x_m}|=1$, the value of $(y_1, y_2, ..., y_m)$ will be uniquely specified and will be a deterministic function of $(x_1, x_2, ..., x_m)$. In this case, coordination reduces to distributed computation. Distributed computation is itself a more general problem that the message transmission problem, since the functions computed by the nodes can be taken to be the message of other nodes.\end{example}

Consider the special case of a network with two nodes with node one has input $X_1$ and node two producing output $Y_2$. We assume that the input of node two, $X_2$, and the output of node one, $Y_1$, are disabled, \emph{i.e.,} $|\mathcal{X}_2|=|\mathcal{Y}_1|=1$. Then for every $x_1$ we have a set $\mathcal{A}_{x_1}\subseteq \mathcal{Y}_2$. Assume a one-way communication link from node one to node two. The goal of the first user will be to send a message from node one to node two that will enable production of $y_2\in\mathcal{A}_{x_1}$ at node two. We show in Section \ref{sec:3} that this  special case of the coordination problem is equivalent with the  ``Hide and Seek" problem of McEliece and Posner \cite{b}.  McEliece and Posner define a source coding problem and a zero-sum ``Hide and Seek" game. Rather surprisingly, they illustrate that the optimal compression rate of the source coding problem can be expressed in terms of  the Nash equilibrium of the game. Additional insight was provided by Lov\'asz who provided an elegant combinatorial argument for the result of McEliece and Posner in \cite{c}. We review other related results and in particular connections with R\'enyi mutual information of order $\alpha$ is discussed in Section \ref{sec:renyi}.

We continue by two examples that illustrate connections with zero error rate distortion  (see  \cite[Ch. 2]{d}, \cite{Max}), and with graph entropy. Firstly, consider a non-negative distortion function satisfying $d(x_1, y_2)=0$ if and only if $y_2\in\mathcal{A}_{x_1}$. Then, coordination is equivalent with zero distortion in reconstruction. Secondly, consider the source coding problem for a source that is taking values in a set $\mathcal X_1$. We are given a graph $\mathcal G$ on $\mathcal X_1$, where two symbols $x_1$ and $x'_1$ are connected to each other if it is legitimate to reconstruct $x'_1$ when the source value is $x_1$. We can model this by assuming that $\mathcal Y_2=\mathcal X_1$ and $\mathcal{A}_{x_1}$ being equal to the set of all $x'_1$ that are connected to $x_1$. The message transmitted from node one to node two represents the compressed message. The minimum compression rate in this case is equal to the logarithm of the chromatic number of the complement of $\mathcal G$ in the one-shot case when only one instance of the source is given. After coloring the vertices of  the complement of $\mathcal G$, the message can be the identity of the color that is assigned to the source symbol.  In the asymptotic case when multiple instances of the source are observed, the answer is the logarithm of the fractional  chromatic number \cite[p. 2215]{zero-error-it}.

For the asymptotically vanishing error model, authors in \cite[Conjecture 1]{cuff}  conjectured that empirical coordination and strong coordination have the same rate regions when infinite common randomness is provided to the parties. Considering a special two node network, we observe connections in the zero error model. In \cite{Winter}, authors considers strong coordination with unlimited common randomness and arrives at expressions that match the one given by McEliece and Posner. However, the work of McEliece and Posner (or its follow up works) are not cited  in \cite{Winter} and the connection is not noted. See Section \ref{sec:strong-coordination} for more details.

This paper is organized as follows: in Section \ref{prelim} we set up the notation that we use. Section \ref{sec:set-coordination} defines set coordination and defines one-way coordination capacity for a two node problem. A general lower bound for this problem is given in Section \ref{sec:general-lower}. Section \ref{sec:3} provides a detailed treatment for the special case of two nodes and discusses its connections with various known results. Section \ref{sec:with-side} computes the coordination capacity when the side information of the second node is a function of the side information of the first node. Finally, in Section \ref{sec:linear}, we consider linear coordination and provide several new results. Extensions to MAC and BC setups are given in Section \ref{sec:extensions}.

\section{Notation and Preliminaries}\label{prelim}
We adopt the notation of \cite{elgamal-kim}. In particular, we show the set $\{1,2,...,m\}$ by $[m]$, and the set $\{k+1, k+2, ..., m\}$ by $[k+1:m]$. All random variables in this paper are finite discrete random variables. All the logarithms are in base 2 in this paper.

Given two graphs $\mathcal G_1, \mathcal G_2$, the tensor product $\mathcal G_1\otimes\mathcal G_2$ is a graph whose vertex set is the Cartesian product of the vertex sets of $\mathcal G_1$ and $\mathcal G_2$ defined as follows: two vertices $(u_1,u_2)$ and $(v_1,v_2)$ in  are adjacent in $\mathcal G_1\otimes\mathcal G_2$  if and only if $u_1$ is adjacent with $v_1$ in $\mathcal G_1$ and $u_2$ is adjacent with $v_2$  in $\mathcal G_1$ .

There are many definitions for R\'enyi mutual information (see \cite{Verdu} for a review). One definition for R\'enyi mutual information of a joint pmf $p(x,y)$ is as follows:
\begin{equation}
I_{\alpha}(X;Y)=\min_{q(y)} D_{\alpha}(p(x,y)||p(x)q(y)),\label{defIalpha}
\end{equation}
Where $D_\alpha$ is the R\'enyi divergence between two pmfs is defined as follows: $$D_\alpha(p\| q) = \frac{1}{\alpha -1} \log \left(\sum_x p(x)^{\alpha}q(x)^{1-\alpha}\right).$$ 
Note that as $\alpha$ converges to one, R\'enyi divergence and R\'enyi mutual information of order $\alpha$ tend to the KL divergence and Shannon's mutual information. 

In \cite[Eq 13]{e}, it is shown that mutual information of order $\alpha$, as defined in equation \eqref{defIalpha}, is equal to:
\begin{equation}
I_{\alpha}(X;Y)=\frac{\alpha}{\alpha-1}\log \bigg(\sum_{y}\Big[\sum_{x}p(x)p(y|x)^\alpha\Big]^{1/\alpha}\bigg).
\end{equation}
R\'enyi mutual information $I_\alpha(X;Y)$ is a non-decreasing function of $\alpha$ for $\alpha\in[0,\infty]$ \cite[Thm 3]{renyi}.  

We now provide explicit expressions for $I_\alpha(X;Y)$ when $\alpha=0$ and $\alpha=\infty$.
When $\alpha$ goes to zero, $\sum_x p(x)p(y|x)^{\alpha}$ rises to power ${1}/{\alpha}$ which goes to infinity; thus only the largest term is important. In fact, one can show that
\begin{align}
\begin{split}
I_0(X;Y)&=\lim_{\alpha \to 0} \frac{\alpha}{\alpha-1}\log \bigg(\max_{y}\Big[\sum_{x}p(x)p(y|x)^\alpha\Big]^{1/\alpha}\bigg)\\&=\lim_{\alpha \to 0}\frac{1}{\alpha-1}\log \Big(\max_{y}\big[\sum_{x}p(x)p(y|x)^\alpha\big]\Big) \\&=-\log(\max_{y}\sum_{p(y|x)>0} p(x)).
\end{split}
\end{align}
Similarly, one can show that as $\alpha$ tends to infinity, we have
\begin{align}I_\infty(X;Y)=\log\left(\sum_{y}\underset{x: p(x)>0}{\max}\;p(y|x)\right)\label{eqn:mutu0inf}.\end{align}
\section{System Model}
\subsection{Set Coordination}\label{sec:set-coordination}
\begin{definition}
Given an input pmf $p(x_1, \cdots, x_m)$ and action sets $\mathcal{A}_{x_1, x_2, ..., x_m}$, one-shot and asymptotic coordination are defined as follows: in one-shot coordination, the parties observe only one instance of $X_i$ and coordination is achieved if $(y_1, y_2, ..., y_m)\in\mathcal{A}_{x_1, x_2, ..., x_m}$ for any $(x_1,\cdots, x_m)$ where $p(x_1, \cdots, x_m)>0$. 
In the asymptotic version, the parties observe $n$ i.i.d.~ repetitions of the sources $X_1^n, X_2^n, \cdots, X_m^n$. Coordination is achieved if  $(y_{1i}, y_{2i}, ..., y_{mi})\in\mathcal{A}_{x_{1i}, x_{2i}, ..., x_{mi}}$  for any $i\in[n]$ and any $(x^n_1,\cdots, x^n_m)$ where $p(x^n_1, \cdots, x^n_m)=\prod_{i=1}^np(x_{1i}, \cdots, x_{mi})>0.$
\end{definition}

We assume that the nodes have access to limited communication resources, as well as possibly private or common randomness. However, similar to empirical coordination, without loss of generality we can assume that the nodes are deterministic and do not use shared or private randomness. This is because  $(y_1, y_2, ..., y_m)\in\mathcal{A}_{x_1, x_2, ..., x_m}$ with probability one, and hence it has to hold for all possible values of the shared randomness variable. 

Set coordination is related to empirical coordination. Take some arbitrary conditional pmf $$p(y_1, y_2, ..., y_m|x_1, x_2, ...,x_m)$$ such that $p(y_1, ..., y_m|x_1, ...,x_m)>0$ only if $(y_1, y_2, ..., y_m)\in\mathcal{A}_{x_1, x_2, ..., x_m}$. Then a zero-error empirical coordination code for $p(y_1, y_2, ..., y_m|x_1, x_2, ...,x_m)$ is also a zero-error set coordination code. However, set coordination is more relaxed in the asymptotic formulation. Take a set coordination code of block length $n$ and two sequence $(x_1^n, x_2^n, \cdots, x_m^n)$ and $(\tilde x_1^n, \tilde x_2^n, \cdots, \tilde x_m^n)$ of the same type. Let $(y_1^n, y_2^n, \cdots, y_m^n)$ and $(\tilde y_1^n, \tilde y_2^n, \cdots, \tilde y_m^n)$ denote the actions of the nodes in response to $(x_1^n, x_2^n, \cdots, x_m^n)$ and $(\tilde x_1^n, \tilde x_2^n, \cdots, \tilde x_m^n)$  respectively. Then it can be the case that the joint types of the sequences $(x_1^n, x_2^n, \cdots, x_m^n, y_1^n, y_2^n, \cdots, y_m^n)$ and $(\tilde x_1^n, \tilde x_2^n, \cdots, \tilde x_m^n, \tilde y_1^n, \tilde y_2^n, \cdots, \tilde y_m^n)$ are different. Therefore, one cannot assign a single empirical conditional type $p(y_1, ..., y_m|x_1, ...,x_m)$ to the set coordination code. 

\begin{remark}
Despite the apparent difference between set coordination and empirical coordination, it would be interesting to study whether set coordination can be expressed in terms of empirical coordination (under either zero error, or asymptotically zero error criteria). As we will see in the proof of Theorem \ref{thm:t1}, the known converse techniques for empirical coordination extend to set coordination. 
\end{remark}

\subsection{One-way coordination capacity}\label{sec:one-waycoordination}
A two nodes network is characterized by two alphabet sets $\mathcal{X}_1$ and $\mathcal{X}_2$ for inputs and two action sets $\mathcal{Y}_1$ and $\mathcal{Y}_2$. For each pair of inputs $(x_1,x_2)\in \mathcal{X}_1\times\mathcal{X}_2$ we have a permissible action set $A_{x_1,x_2}\subseteq \mathcal{Y}_1\times\mathcal{Y}_2$. We are given some $p(x_1, x_2)$ on the inputs. Assume that there is a one-way communication link of limited rate $R$ from node one to node two, as depicted in Fig.~\ref{fig:one-way}. 
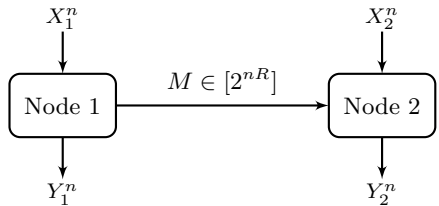
\begin{figure}
\begin{center}
\begin{tikzpicture}[node distance=3.5cm,auto,>=latex', scale=1.4]
    \draw[black, thick,rounded corners] (0,0) rectangle (1,0.6) node[font=\fontsize{9.5}{144}\selectfont,pos=.5,anchor=center] {Node 1};
	\draw[black, thick,->] (1,.3) -- (3,.3) node[font=\fontsize{9}{144}\selectfont,pos=.5,anchor=south] {$M\in [2^{nR}]$}; 
	\draw[black,thick,rounded corners](3,0) rectangle (4,0.6) node[font=\fontsize{9.5}{144}\selectfont,pos=.5,anchor=center]{Node 2}; 
	\draw[black, thick,->] (0.5,1) ->(0.5,0.6) node[font=\fontsize{8.5}{144}\selectfont,pos=.2,anchor=south]{$X^n_1$};
	\draw[black, thick,->] (0.5,0) -> (0.5,-0.4) node[font=\fontsize{8.5}{144}\selectfont,pos=.85,anchor=north] {$Y^n_1$};
	\draw[black, thick,->] (3.5,1) -> (3.5,0.6) node[font=\fontsize{8.5}{144}\selectfont,pos=.2,anchor=south] {$X^n_2$};
	\draw[black, thick,->] (3.5,0) -> (3.5,-.4) node[font=\fontsize{8.5}{144}\selectfont,pos=.85,anchor=north] {$Y^n_2$};
\end{tikzpicture}
\end{center}
\caption{Coordination using a one-way communication link.}\label{fig:one-way}
\end{figure}

\begin{definition}
Coordination is achievable with one-way communication rate $R$ with block length $n$ if there are encoding and decoding functions
\begin{align}
\nonumber
&\mathcal{E}:{\mathcal{X}_1}^n\mapsto[2^{nR}],\qquad\mathcal{D}_1:\mathcal{X}_1^n\mapsto \mathcal{Y}^n_1,\qquad \mathcal{D}_2:[2^{nR}]\times \mathcal{X}_2^n\mapsto \mathcal{Y}^n_2,
\end{align}
such that $(Y_1^n, Y_2^n)=(\mathcal{D}_1(X_1^n), \mathcal{D}_1(M, X_2^n))$ is coordinated with $(X_1^n, X_2^n)$ where the message $M=\mathcal{E}(X_1^n)$.

A rate $R$ is said to be one-shot achievable if it is achievable with a code of block length $n=1$, and is said to be asymptotically achievable if  it is achievable with a code with for some arbitrarily large block length $n$. We use $C$ to denote the maximum one-shot achievable rate, and $\bar{C}$ to denote the supremum of the  asymptotically achievable rates.
\end{definition}
\begin{remark}Computing the  one-way capacity $\bar C$ is in general a difficult problem. For instance, consider the special case of $Y_1$ being a constant random variable, i.e. $|\mathcal{Y}_1|=1$, and $\mathcal{Y}_2=\mathcal{X}_1$. Further assume $y_2\in A_{x_1, x_2}$ if and only if $y_2=x_1$; in other words, the task of node one will be to communicate $x_1$ to node two. Given some $p(x_1, x_2)$, the problem reduces to a zero-error version of the Slepian-Wolf problem. The minimum rate for  zero-error source coding with side information is known as Witsenhausen's zero-error rate and is open in general \cite{Witsenhausens}. Practical code designs for this problem can be found in \cite{Slepian1, Slepian2, Slepian3}.
\end{remark}

\section{A general lower bound}\label{sec:general-lower}
One can use the fact that zero-error coordination is a more stringent condition than vanishing error coordination to derive the following lower bound on the asymptotic coordination capacity. Let $\mathcal{P}$ be the class of all $p(y_1,y_2|x_1,x_2)$ such that $p(y_1,y_2|x_1,x_2)>0$ only if $(y_1, y_2)\in A_{x_1, x_2}$. Then we have:
\begin{theorem}\label{thm:t1}
If node one observes $X_1$ and node two observes $X_2$ and they want to produce $Y_1$ and $Y_2$ respectively then minimum rate required from node one to node two is bounded from below by:
\begin{equation}
\max_{\substack{q(x_1,x_2): \\q(x_1, x_2)=0 \text{ if }p(x_1, x_2)=0}} ~\min_{q(y_1,y_2|x_1,x_2)\in \mathcal{P}} \min_{\substack{ q(f|x_1,y_1): \\I_q(X_1,Y_1;Y_2|F,X_2)=0}} I(F;X_1|X_2).
\end{equation}
\end{theorem}
\begin{proof}[Proof of Theorem \ref{thm:t1}] Coordination capacity depends only on the support of $p(x_1, x_2)$, i.e.~$\{(x_1, x_2): p(x_1, x_2)>0\}$. Furthermore, if the support of a pmf $q(x_1, x_2)$ is smaller than the support of $p(x_1, x_2)$, the coordination capacity for $q(x_1, x_2)$ will be less than or equal to that for $p(x_1, x_2)$. Therefore,  it suffices to show that the coordination capacity for the input pmf $p(x_1, x_2)$ is bounded from below by
\begin{equation}
\min_{p(y_1,y_2|x_1,x_2)\in \mathcal{P}} \min_{\substack{p(f|x_1,y_1): \\I_p(X_1,Y_1;Y_2|F,X_2)=0}} I(F;X_1|X_2).
\end{equation}
Take an arbitrary set coordination code of length $n$. Consider the  converse given in \cite[Appendix G]{Yassaee} for interactive channel simulation under the empirical coordination constraint. Consider the special case of one round of communication $r=1$. This converse can be exactly mimicked for the set coordination code until equation (102). In equations  (102) onwards in \cite[Appendix G]{Yassaee}, it is argued that the joint pmf of $(Y_{1Q}, Y_{2Q}, X_{1Q}, X_{2Q})$, where $Q$ is the time sharing random variable, is close to the desired joint pmf $p(y_1,y_2,x_1,x_2)$ that we wish to simulate. By contrast, in set coordination, we have that $(Y_{1Q}, Y_{2Q})\in A_{X_{1Q}, X_{2Q}}$ with probability one. Therefore, $p(y_{1Q}, y_{2Q}|x_{1Q}, x_{2Q})\in\mathcal{P}$. Therefore, the last step in the proof can be completed, and we get that the transmission rate is greater than or equal to $I(F;X_{1Q}|X_{2Q})$. Since  $p(y_{1Q}, y_{2Q}|x_{1Q}, x_{2Q})\in\mathcal{P}$, we get that the coordination rate is greater than or equal to
\begin{equation}
\min_{p(y_1,y_2|x_1,x_2)\in \mathcal{P}} \min_{\substack{p(f|x_1,y_1): \\I_p(X_1,Y_1;Y_2|F,X_2)=0}} I(F;X_1|X_2).
\end{equation}
\end{proof}
\begin{corollary}\label{cor} If $X_2$ and $Y_1$ are constant random variables, random variable $F$ will satisfy $X_1\rightarrow F\rightarrow Y_2$ and we would like to minimize $I(X_1;F)$. The choice of $F=Y_2$ will be optimal and we get that the coordination capacity is greater than or equal to 
$$\max_{q(x_1)}\min_{p(y_2|x_1)\in \mathcal{P}}I(X_1;Y_2).$$
We will see later that this lower bound is tight.
\end{corollary}

 \section{No side information at node two}\label{sec:3}
In this section we restrict to the special case of a two nodes network, with node two having no observation, \emph{i.e.,} $|\mathcal{X}_2|=1$. Without loss of generality we assume that $p(x_1)>0$ for all $x_1$ throughout this section. In this case we can make some simplifications as follows: we can drop $x_2$ from $A_{x_1x_2}$ and define $$A_{x_1}\triangleq \{y_2:\exists y_1\in \mathcal{Y}_1, (y_1,y_2)\in A_{x_1x_2}\}.$$ Notice that if node two chooses a proper action $y_2$ from $A_{x_1}$, by definition there will exist a proper action $y_1$ for node one.  Since node one knows $x_1$ and the message sent to node two as well as the decoding strategy of node two, action $y_2$ and hence $y_1$ can be found by this node. 

\subsection{One-shot achievable rates}\label{one-shot}
We would like to compute the minimum size of the alphabet set of message $M$ sent from node one to node two, such that node two can choose an action $y_2\in\mathcal{Y}_2$ where $y_2 \in A_{x_1}$. In other words, assuming that $M\in \mathcal{M}$, we would like to minimize $|\mathcal{M}|$ as much as possible. The one-shot capacity is minimum of $\log|\mathcal{M}|$.

\textbf{Characterization in terms of graph covers:}
The one-shot capacity can be also expressed in terms of the size of minimum graph cover \cite{b}. Before proceeding, we need some definitions:

\begin{definition}[Coordination Graph]\label{coordinationgraph}
 Consider a bipartite graph $\mathcal G=(\mathcal X_1, \mathcal Y_2)$, where by this notation we mean that the vertices in one part being indexed by elements of $\mathcal X_1$ and the vertices of the other part indexed by elements of $\mathcal Y_2$. An edge is drawn between $x_1\in \mathcal X_1$ and $y_2\in \mathcal{Y}_2$ if and only if $y_2\in A_{x_1}$. We call this a \emph{coordination graph}. Both the one-shot and asymptotic coordination capacity are characterized by the coordination graph. Hence we sometime denote the one-shot and asymptotic coordination capacity by $C(\mathcal G)$ and $\bar C(\mathcal G)$ respectively.
\end{definition}

\begin{definition}[\cite{b}]
Consider a bipartite graph $\mathcal{G}=(\mathcal X_1, \mathcal Y_2)$. A cover set for $\mathcal X_1$ is a subset $\mathcal S\subseteq \mathcal Y_2$ such that for each $x_1\in \mathcal X_1$, there exist some $y_2\in \mathcal S$ such that $x_1$ and $y_2$ are connected.
\end{definition}
Observe that the minimum length of message from node one to node two, $|\mathcal{M}|$ is equal to the size of minimum cover for $\mathcal X_1$. This is because given a cover set $\mathcal S$, node one can simply send the index of $y_2$ in the cover set $\mathcal S$ that is connected to $x_1$. Conversely, given any strategy by node one, one can produce a cover set by putting together the outputs $y_2$ corresponding to different values of the message $M$. 

Finding minimum cover set for graph $\mathcal{G}=(\mathcal X_1, \mathcal Y_2)$ is equivalent to solving following integer linear programming: 
\begin{align}
\begin{split}
 \mathsf{IP} (\mathcal G)&= \min \sum_{y_2\in \mathcal{Y}_2} \alpha(y_2)  \\   
 s.t. \hspace{3 mm} \forall x_1&\in \mathcal{X}_1 \hspace{2 mm} \sum_{y_2: y_2\sim x_1} \alpha(y_2) \geq 1 \hspace{3 mm} \\  &\alpha(y_2)\in \{0,1\}.
\end{split}\label{eqn:LP}
\end{align} 
Here $C(\mathcal G)=\log_2 \mathsf{IP} (\mathcal G)$ is the minimum message length in terms of bits.
\subsection{Asymptotic achievable rates}
In the asymptotic version of the problem, we fix some block length $n$. The first node observes some sequence $x_1^n$ and needs to convey it to node two in a way it can produce $y_2^n$ such that $y_{2i}\in A_{x_{1i}}$ for $i\in [n]$. Then, one can see that the solution to this problem is size of minimum cover for a bipartite graph on $\mathcal{X}_1^n$ and $\mathcal{Y}_2^n$, with two sequences $x_1^n$ and $y_2^n$ connected to each other if and only if $y_{2i}\in A_{x_{1i}}$ for $i\in [n]$. This graph can be expressed as $\mathcal{G}^n$, \emph{i.e., }the tensor product of the graph for the one-shot case, $\mathcal{G}$, with itself by $n$ times. Then the minimum required rate for block length $n$ is $\log (\mathsf{IP}(\mathcal G^{\otimes n}))/n$, and the limit of this when $n$ goes to infinity is the asymptotic coordination capacity is equal to $\bar C(\mathcal G)$, \emph{i.e.,}
\begin{equation}
\bar C(\mathcal G)\triangleq \log\lim_{n \to \infty}\big( \mathsf{IP}(\mathcal G^{\otimes n})^{1/n}\big).
\end{equation}  
The above limit exist by the Fekete's lemma because the sequence of $\log (\mathsf{IP}(\mathcal G^{\otimes n}))$ is superadditive \cite{b}. In fact, for any two arbitrary graphs $\mathcal G_1$ and  $\mathcal G_2$, we have
\begin{align}
\label{prodID}
\mathsf{IP}(\mathcal G_1\otimes \mathcal G_2)\leq \mathsf{IP}(\mathcal G_1)\mathsf{IP}(\mathcal G_2),
\end{align}
because one possible way to find a cover set for $\mathcal G_1\otimes \mathcal G_2$ is to first find a cover set for $\mathcal G_1$ and for $\mathcal G_2$, and then take the Cartesian product of these two cover sets as a cover set for $\mathcal G_1\otimes \mathcal G_2$.

Authors in \cite{b} also consider the above asymptotic version of the problem and show that $\bar C(\mathcal G)=\log_2\mathsf{LP}(\mathcal G)$,  where $\mathsf{LP}(\mathcal G)$ is the relaxation of linear programming given in equation \eqref{eqn:LP} as follows:
\begin{align}
\begin{split}
\mathsf{LP}(\mathcal G)&= \min \sum_{y_2\in \mathcal{Y}_2} \alpha(y_2)  \\   
 s.t. \hspace{3 mm} \forall x_1&\in \mathcal{X}_1 \hspace{2 mm} \sum_{y_2: y_2\sim x_1} \alpha(y_2) \geq 1 \hspace{3 mm} \\  &\alpha(y_2)\geq 0.
\end{split}\label{eqn:LP2}
\end{align} 
Motivated by the Hide and Seek problem of \cite{b}, Lov\'asz provides a combinatorial proof for the above relation in \cite{c}. In order to prove \begin{align}\bar C(\mathcal G)=\log\lim_{n \to \infty}\big( \mathsf{IP}(\mathcal G^{\otimes n})^{1/n}\big)=\log\mathsf{LP}(\mathcal G),\label{LovaszThem}\end{align} firstly, it is clear that the above is relaxation of the linear program given in equation \eqref{eqn:LP}. Therefore $\mathsf{LP}(\mathcal G)\leq  \mathsf{IP}(\mathcal G)$. By using the dual problem of linear programming, Lov\'asz shows that \begin{align}
\label{prodID2}\mathsf{LP}(\mathcal G_1\otimes \mathcal G_2)=\mathsf{LP}(\mathcal G_1)\mathsf{LP}(\mathcal G_2).\end{align} Thus $\mathsf{IP}(\mathcal G^{\otimes n})\geq \mathsf{LP}(\mathcal G^{\otimes n})=\mathsf{LP}(\mathcal G)^n$ for all $n$. Hence $\bar C(\mathcal G)\geq \log(\mathsf{LP}(\mathcal G))$. To show the other direction, $\log(\mathsf{LP}(\mathcal G))\geq \bar C(\mathcal G)$,  Lov\'asz provides an elegant combinatorial argument using a greedy algorithm and the dual of the linear program. We refer the readers to \cite{c} for the details.

\begin{remark}
While computing the minimum cover $C(\mathcal G)$ is an NP-complete problem (it is the sixth problem of Karp's 21 NP-complete problems \cite{NP}), $\bar C(\mathcal G)$ can be calculated in polynomial time, as it is the solution of a real linear program.

It is known that $C(\mathcal G)$ can be arbitrarily large while $\bar C(\mathcal G)$ is arbitrarily small \cite{Orl}. This happens for the mail-order problem of Slepian, Wolf, and Wyner \cite{Mail}. On the other hand, 
$C(\mathcal G)=\bar C(\mathcal G)$ holds for two families of ``interval graphs and forests" that are defined in \cite{Max}. 

\end{remark}

We now make the observation that this special case of the coordination problem is nothing but the ``Hide and Seek"  problem of McEliece and Posner \cite{b}. Consider the following source coding problem. Let $\mathcal{X}_1$ be a set of natural numbers. Alice observes some number $x_1\in\mathcal{X}_1$. We have a certain list of properties such as a number being even, being divisible by five, being a prime number, etc. Given some $x_1\in\mathcal{X}_1$, Alice can find a subset of properties that are satisfied by $x_1$, \emph{e.g.,} if $x_1=5$, it is both prime and divisible by 5. The goal of Alice is to inform Bob of \emph{at least one} of these valid properties; thus the goal is not to inform Bob of $x_1$, but one of its valid properties. The question is the minimum amount of communication needed from Alice to Bob to accomplish this task. It is not difficult to see that this problem is identical to the two node coordination problem with no side information at node two: $\mathcal{Y}_2$ can denote the set of properties and $A_{x_1}$ can contain the list of properties that $x_1$ has.

\textbf{Connection to game theory:}
The solution to the Hide and Seek problem can be expressed in terms of the Nash equilibrium of the following zero-sum game: let player one (hider) choose $x_1\in\mathcal{X}_1$ and player two (seeker) choose some $y_2\in \mathcal{Y}_2$. Player one has to give player two one dollar if $y_2\in \mathcal A_{x_1}$, otherwise the payoff is zero. This game is called Hide and Seek because player one hides and player two seeks player one. We refer the readers to \cite{b} for details. 

The fact that the above source coding is related to this game may come as a surprise. This is due to the fact that the commonly used achievability proofs and converses in information theory are not based on the  Nash value of games. However, we observe that 
\begin{itemize}
\item The key feature of capacity regions and upper bounds thereof are essentially \emph{additive} regions, meaning that they expand by a factor of $n$, when evaluated on $n$ independent repetitions of a problem (see \cite{duality} for a discussion). The Nash value of the above zero-sum repeated game also has the additivity property, as the expected value of the total payoff is equal to the sum of the expected value of the payoffs in the individual games. Therefore, it is quite possible that rates of a capacity region have characterizations in terms of Nash value of carefully constructed games.
\item One can draw simple \emph{operational} connections between the communication problem and the Hide and Seek game. Take an arbitrary code in the communication problem. In the communication problem, the seeker (player two) can ensure that $y_2\in \mathcal A_{x_1}$ if he receives $R$ bits from the hider. However, in the game setup, there is no communication link. Nonetheless, the seeker can still guess the $R$ message bits and win the game with probability at least $2^{-R}$, regardless of the value of $x_1$. Therefore, the seeker has a strategy that gives him a payoff of $2^{-R}$ regardless of the action of the hider. This gives a lower bound on the Nash value of the game in terms of minimum value of $R$. This lower bound is tight by the result of  \cite{b}.
\end{itemize}

\subsubsection{Connection to the Rate Distortion theory}
While elegant, Lov\'asz's proof is combinatorial. Fortunately, the asymptotic capacity $\bar C$ can be found using standard information theory arguments as in \cite{b}. A formal way to do so is to express the the problem in terms of a zero-error rate distortion problem. Consider a distortion measure $d:\mathcal{X}_1\times\mathcal{Y}_2\to\{0,1\}$ where $d(x_1,y_2)=0$ if and only if $x_1$ and $y_2$ be connected, \emph{i.e.,} $y_2\in A_{x_1}$. When  $x_1$ and $y_2$ are not connected, we can set the distortion to an arbitrary positive value, say $d(x_1,y_2)=1$.

Let $R(p(x_1), D)$ be the standard rate distortion function for distortion function $d(\cdot, \cdot)$ when the source $X_1$ has pmf $p(x_1)$. Furthermore, let $R_0(p(x_1), D)$ be the zero-error rate distortion function, which is the minimum (asymptotic) rate which can guaranty average distortion less than or equal to $D$ \emph{with probability one}, \emph{i.e., } with probability of excess distortion being zero. Then, it is easy to see that $R_0(p(x_1), D)$ at $D=0$ is equal to $\bar{C}_\mathcal{G}$ because $E[d(X_1^n,Y_2^n)]\leq 0$ implies that then $d(x_1^n,y_2^n)=0$ for all pairs $(x_1^n,y_2^n)$ where $p(x_1^n,y_2^n)>0$. 

Observe that $R_0(p(x_1), D)$ depends only on the support of $p(x_1)$, \emph{i.e.,} the values of $x_1$ where $p(x_1)>0$, and not on the exact values of probabilities $p(x_1)$. Furthermore, it is clear that $R(p(x_1), D)\leq R_0(p(x_1), D)$ since in $R_0(p(x_1), D)$ we ask for exactly zero probability of exceeding the distortion, whereas in $R(p(x_1), D)$ we ask for an asymptotically vanishing probability of excess distortion. Therefore, assuming that $p(x_1)>0$ for all $x_1$, we have
\begin{equation*}
R_0(p(x_1), D)\geq \max_{q(x_1)}R(q(x_1), D).
\end{equation*}
Interestingly, the above inequality holds with equality \cite[Thm 4.2]{d}:
\begin{equation}
R_0(p(x_1), D)=\max_{q(x_1)}R(q(x_1), D)=\max_{q(x_1)}~\min_{p(y_2|x_1): \mathbb E d(X_1,Y_2)\leq D} I(X_1;Y_2).\label{eqn:43}
\end{equation}

Now, let us specialize this result to our coordination problem, when distortion $D=0$. Let $$\mathcal{P}=\{p(y_2|x_1): p(y_2|x_1)>0 \textit{ only if }y_2\in A_{x_1}\}.$$ Then, the asymptotic coordination capacity is equal to \begin{equation}
\bar{C}(\mathcal G)=\max_{q(x_1)}\min_{p(y_2|x_1)\in \mathcal{P}}I(X_1;Y_2).
\end{equation} 

An alternative expression for $\bar{C}(\mathcal G)$ follows from \cite[Cor. 3.7]{d}, where it is shown that the standard rate-distortion function has the following characterization at zero distortion $D=0$:
\begin{equation}
R(p(x_1), 0)=-\min_{q(x_1)}\Big\{D(p(x_1)||q(x_1))+\max_{y_2\in \mathcal{Y}_2}\log\big(\sum_{x_1:d(x_1,y_2)=0}q(x_2)\big)\Big\}.
\end{equation}
Then from equation \eqref{eqn:43}, we have 
\begin{align}
\bar{C}(\mathcal G)&=R_0(p(x_1), 0)\nonumber
\\&=-\min_{q(x_1)}\max_{y_2\in \mathcal{Y}_2}\log\big(\sum_{x_1:d(x_1,y_2)=0}q(x_1)\big)\nonumber
\\&=\max_{q(x_1)}-\max_{y_2\in \mathcal{Y}_2}\log\big(\sum_{x_1:d(x_1,y_2)=0}q(x_1)\big)\label{eqn:re1}
\\&=\max_{q(x_1)}\min_{p(y_2|x_1)\in \mathcal{P}}I_0(X_1;Y_2),\nonumber
\end{align}
Where $I_0(X_1;Y_2)$ is the R\'enyi mutual information of order zero, defined in Section \ref{prelim}. The last equality holds because for any $q(x_1)$, from the definition of  $I_0(X_1;Y_2)$ we have that $$\min_{p(y_2|x_1)\in \mathcal{P}}I_0(X_1;Y_2)=-\log\big(\max_{y_2\in \mathcal{Y}_2}\sum_{x_1: d(x_1, y_2)=0}q(x_1)\big).$$
\begin{remark}
Curiously,  equation \eqref{eqn:re1} is also equal to the zero-error feedback capacity of a point-to-point channel, when it is positive. This expression is also related to the sphere-packing bound for error exponents (see for instance \cite[Eq. (8)]{Dalai}). It would be interesting to find operational interpretations for these facts.
\end{remark}

\subsubsection{Characterization in terms of R\'enyi mutual information}\label{sec:renyi}
So far, we have mentioned three characterizations for $\bar{C}(\mathcal G)$: one in terms of a relaxed linear program by Lov\'asz, and two as follows:
\begin{align}
\bar{C}(\mathcal G)&=\max_{q(x_1)}\min_{p(y_2|x_1)\in \mathcal{P}}I_0(X_1;Y_2), \label{eqn:mm1}
\\&=\max_{q(x_1)}\min_{p(y_2|x_1)\in \mathcal{P}}I(X_1;Y_2).  \label{eqn:mm2}
\end{align}
It is desirable to provide connections between various characterizations. Our first result states that the expression of Lov\'asz's linear program can be also understood in terms of R\'enyi mutual information, which to best of our knowledge is new. In particular, we use that Lov\'asz's linear program to show the following theorem:
\begin{theorem}Assuming $p(x_1)>0$ for all $x_1$, we have \label{thm:thm1}
\begin{align}
\bar{C}(\mathcal G)&=\max_{q(x_1)}\min_{p(y_2|x_1)\in \mathcal{P}}I_\infty(X_1;Y_2). \label{eqn:mm3}
\end{align}
\end{theorem}
\begin{remark} Since  mutual information of order $\alpha$ is a non-decreasing function in $\alpha$, equations \eqref{eqn:mm1}-\eqref{eqn:mm3} imply that for any $\alpha\in[0,\infty]$ we have
\begin{align}
\bar{C}(\mathcal G)&=\max_{q(x_1)}\min_{p(y_2|x_1)\in \mathcal{P}}I_\alpha(X_1;Y_2). \label{eqn:mm4}
\end{align}
\end{remark}
\begin{proof}[Proof of Theorem \ref{thm:thm1}]
Using equation \eqref{eqn:mutu0inf}, we need to show that
\begin{align*}\bar{C}(\mathcal G)&= \max_{q(x_1)}\min_{p(y_2|x_1)\in \mathcal{P}}\log\left(\sum_{y_2}\underset{x_1: q(x_1)>0}{\max}\;p(y_2|x_1)\right)
\\&=\min_{p(y_2|x_1)\in \mathcal{P}}\log\left(\sum_{y_2}\underset{x_1}{\max}\;p(y_2|x_1)\right).\end{align*}
Let $\alpha({y_2})=\underset{x_1}{\max}\;p(y_2|x_1)$. Then, we would like to show that 
\begin{align*}\bar{C}(\mathcal G)= \log\min_{p(y_2|x_1)\in \mathcal{P}}\sum_{y_2}\alpha({y_2}).\end{align*}
Thus, $\bar{C}(\mathcal G)=\log \mathsf{LP}_2(\mathcal G)$ where
\begin{align}
\begin{split}
\mathsf{LP}_2(\mathcal G)&= \min \sum_{y_2} \alpha(y_2)  \\   
 s.t. \hspace{3 mm} \alpha(y_2)&\geq p(y_2|x_1),\qquad\forall x_1
\\\hspace{3 mm}\sum_{y_2} p(y_2|x_1)&=1, \qquad p(y_2|x_1)\geq 0
\\\hspace{3 mm}p(y_2|&x_1)=0 \text{ if } y_2\nsim x_1,
\end{split}\label{eqn:LP3n}
\end{align} 
where by $y_2\nsim x_1$ we mean that $y_2\notin A_{x_1}$.
Now, recall Lov\'asz's linear programming formulation: $\bar{C}(\mathcal G)=\log \mathsf{LP}(\mathcal G)$ where 
\begin{align}
\begin{split}
\mathsf{LP}(\mathcal G)&= \min \sum_{y_2} \alpha(y_2)  \\   
 s.t. \hspace{3 mm} \forall x_1&\in \mathcal{X}_1 \hspace{2 mm} \sum_{y_2: y_2\sim x_1} \alpha(y_2) \geq 1 \hspace{3 mm} \\  &\alpha(y_2)\geq 0.
\end{split}\label{eqn:LP2n}
\end{align} 
Thus, we only need to show that the two LPs are equivalent. The LP given in equation \eqref{eqn:LP3n} has more variables than the one given in equation \eqref{eqn:LP2n}. Using a standard Fourier-Motzkin elimination on variables $p(y_2|x_1)$, one can see that the former LP reduces to the latter LP. Hence we are done.

\end{proof}

Our aim was to show the connection between Lov\'asz's linear program and the R\'enyi mutual information. But it is possible to provide an algebraic proof for  Theorem \ref{thm:thm1} (see \cite{Max} for another algebraic proof).

\begin{proof}[Second proof of Theorem \ref{thm:thm1}] Let $$F_\alpha(q(x_1))= \min_{p(y_2|x_1)\in \mathcal{P}}I_\alpha(X_1;Y_2).$$ Since R\'enyi mutual information $I_\alpha(X;Y)$ is a non-decreasing function of $\alpha$, we have that $$F_1(q(x_1))\leq F_\infty(q(x_1))\qquad \forall q(x_1).$$
Furthermore, using the fact that $p(x_1)>0$ for all $x_1$ we have that
$$F_\infty(q(x_1))= \log\left(\sum_{y_2}\underset{x_1}{\max}\;p(y_2|x_1)\right)$$
does not depend on $q(x_1)$. Therefore, it suffices to show that at the $q(x_1)$ that maximizes $F_1(q(x_1))$, we have that $F_1(q(x_1))\geq F_\infty(q(x_1))$. 

A theorem by Shannon \cite[Thm. 2]{ze} studies the pmf $q(x_1)$ that maximize $F_1(q(x_1))$. This theorem implies that the joint pmf that obtains $\max_{q(x_1)}F_1(q(x_1))$ satisfies $q(x,y)=c\cdot q(x)q(y)$ for some fixed $c$ when $q(x,y)>0$. This implies that $I(X;Y)=I_{\infty}(X;Y)=\log(c)$. Hence,
$$ F_\infty(q(x_1))= \min_{p(y_2|x_1)\in \mathcal{P}}I_\infty(X_1;Y_2)\leq \log(c)=\max_{q(x_1)}F_1(q(x_1)).$$
\end{proof}

\subsection{Connections with strong coordination}\label{sec:strong-coordination}
Take some arbitrary $p(y_2|x_1)$. Then, consider the following channel simulation problem: Alice observes a sequence $x_1^n$, unknown to Bob. The goal of Bob is to sample a sequence $y_2^n$ from the pmf $\prod_{i=1}^np(y_{2i}|x_{1i})$. To achieve this, Alice can send $nR$ noiseless bits to Bob, \emph{i.e.,} the communication rate is $R$. This is the problem of  simulating the memoryless channel $p(y_2|x_1)$ via a noiseless link of limited rate using the strong coordination criterion. Alice and Bob may share common randomness, independent of $x_1^n$, at some limited rate $R_0$. 

The channel simulation problem implies set coordination if  $p(y_2|x_1)$ is such that $p(y_2|x_1)>0$ only if $y_2\in {A}_{x_1}$. As discussed in the introduction, strong coordination with infinite shared randomness $R_0=\infty$ is related to the empirical coordination (at least in the vanishing error formulation).  It turns out that the connection between strong and empirical coordiantion exists here in the zero error case as well. Observe that  Empirical coordination, itself, is related to set coordination.

The above zero-error strong coordination problem has been studied in \cite{Winter}, where it is shown that the minimum rate $R$, when $R_0=\infty$, is equal to  \cite[Theorem 24]{Winter}: $$\log(\sum_{y_2} \max_{x_1} p(y_2|x_1)).$$
Even though not mentioned by the authors of \cite{Winter}, the above relation is nothing but $I_\infty(X_1;Y_2)$ for a $p(x_1)$ where $p(x_1)>0, \forall x_1$. Also in Section \cite[III.G]{Winter} (the section on \emph{Weak simulation and reversibility}), formulas similar to the ones given in the first proof of Theorem \ref{thm:thm1}  (both of the linear programs) appear, even though Theorem \ref{thm:thm1} is considering a different problem. Although, the linear program of equation \eqref{eqn:LP2} appears but the works of McEliece and Posner, or Lov\'asz are not cited.

\section{Two nodes with side information}\label{sec:with-side}

Assume that node two has some inputs $X_2$, but that node one is aware of that, \emph{i.e., }$H(X_2|X_1)=0$. 
As in the previous section, we can ignore $Y_1$ in our analysis. Let us define $$A_{x_1,x_2}\triangleq \{y_2:\exists y_1\in \mathcal{Y}_1, (y_1,y_2)\in A_{x_1x_2}\}.$$  If node two chooses a proper action $y_2$ from $A_{x_1,x_2}$, by definition there will exist a proper action $y_1$ for node one.  Since node one knows $x_2$ and the message sent to node two as well as the decoding strategy of node two, action $y_2$ and hence $y_1$ can be found by this node. 

Without loss of generality we assume that $p(x_2)>0$ for all $x_2$, throughout this section.

\begin{definition}Take some arbitrary $x^*_2\in\mathcal{X}_2$ and consider the coordination problem where $X_2=x_2^*$ is fixed and known to everybody, \emph{i.e.,} instead of the joint pmf $p(x_1, x_2)$ we consider the joint pmf $q(x_1, x_2)=p(x_1|x_2^*)\textbf{1}[x_2=x_2^*]$. Since $X_2$ is assumed to be fixed, this problem falls into the class of problems considered in Section \ref{sec:3}, and using Definition \ref{coordinationgraph}, a bipartite graph can be associated to it. We denote this graph by $\mathcal G_{x_2^*}$, which defined on $(\tilde{\mathcal{X}_1}, \mathcal{Y}_2)$, where $\tilde{\mathcal{X}_1}=\{x_1: p(x_1|x_2^*)>0\}$.
\end{definition}

\iffalse
Assume that $\mathcal{X}_1=\{x_{1,1},\cdots,x_{1,p}\}$ and $\mathcal{X}_2=\{x_{2,1},\cdots,x_{2,m}\}$ and $\mathcal{Y}_2=\{y_1,y_2,\cdots,y_k\}$. Now let $\mathcal G=(\mathcal X,\mathcal Y_2)$ notice that the vertices of one part are indexed by elements of $\mathcal X=\mathcal X_1\times \mathcal X_2$. Now for a specific $s \in \mathcal{X}_2$, $\mathcal G_s=(\mathcal X',\mathcal Y_2)$ denote a sub-graph of $G$ where the second component of inputs are fixed and are equal to $s$.

In one shot case, if $X_2=s_i$ then by result of Section \rom{3} we need $\log (C(\mathcal G_{s_i}))$ bits, and since we have to support all inputs of node two $\max_{s} \log (C(G_s))$ is required.
\fi

\begin{theorem}\label{thm:side-inf}
The one-shot coordination capacity for joint pmf $p(x_1, x_2)$ where $H(X_2|X_1)=0$ is equal to 
$$\max_{x_2}C(\mathcal G_{x_2}).$$
The asymptotic coordination capacity for joint pmf $p(x_1, x_2)$ where $H(X_2|X_1)=0$ is equal to 
$$\max_{x_2}\bar C(\mathcal G_{x_2}).$$
\end{theorem}
\begin{remark}
This result implies a cut-set bound  for general coordination networks. If we divide the set of nodes into two groups, say, nodes $\{1,2,..., k\}$ in one group and nodes $\{k+1, k+2, ..., m\}$ in the second group, we can write a cut-set bound as follows: we assume two super-nodes that have access to $x_{[k]}$  and $x_{[k+1:m]}$, and need to make actions $y_{[k]}$  and $y_{[k+1:m]}$ respectively. We also assume that a genie provides the inputs of the second group $x_{[k+1:m]}$ to the first super-node. Then, the minimum total communication rate from super-node one to super-node two is bounded from below by the corresponding bound given in the above theorem for such a two node scenario.
\end{remark}

\begin{proof}[Proof of Theorem \ref{thm:side-inf}] The equation for one-shot case follows directly from the definition of $\mathcal G_{x_2}$, the fact that coordination needs to hold for all values of $x_2$, and $x_2$ is known by both the nodes. 

It remains to show the result for the asymptotic case. Using the one-shot result and applying it to the $n$-letter version of the problem, the minimum coordination rate for codebooks of length $n$, will be equal to 
$$\frac{1}{n}\max_{x^n_2}C(\mathcal G_{x^n_2}),$$
where $\mathcal G_{x^n_2}=\mathcal G_{x_{21}}\otimes\mathcal G_{x_{22}}\otimes...\otimes\mathcal G_{x_{2n}}$ is the tensor product of the graphs for indices $1$ to $n$. Therefore, the one-way communication capacity is equal to
\begin{equation}
\lim_{n\to \infty} \frac{1}{n}\max_{x_2^n} C(\mathcal G_{x_2^n}).
\end{equation}
Thus, we need to show that 
$$\lim_{n\to \infty} \frac{1}{n}\max_{x_2^n} C(\mathcal G_{x_2^n})=\max_{x_2}\bar C(\mathcal G_{x_2}).$$
Using the fact that $\log C(\mathcal G)=\mathsf{IP}(\mathcal G)\geq \mathsf{LP}(\mathcal G)=\log\bar C(\mathcal G)$ and $\mathsf{LP}(\mathcal G_1\otimes \mathcal G_2)=\mathsf{LP}(\mathcal G_1)\mathsf{LP}(\mathcal G_2)$, we have
\begin{align}
\begin{split}
 \frac{1}{n}\max_{x_2^n} C(\mathcal G_{x_2^n}) &\geq \log\max_{x_2^n} \mathsf{LP}(\mathcal G_{x_2^n})^{1/n}\\
&= \log\max_{x_2^n}  \prod_{i=1}^{n} \mathsf{LP}(\mathcal G_{x_{2i}})^{1/n}\\
&= \log \max_{x_2}\mathsf{LP}(\mathcal G_{x_{2}})\\
&=\max_{x_2}\log\mathsf{LP}(\mathcal G_{x_{2}})\\
&=\max_{x_2}\bar{C}(\mathcal G_{x_{2}}).
\end{split}
\end{align}

For the reverse direction we prove that for a given $\epsilon>0$, there exist a natural number $N$ such that for all $n>N$, we have
\begin{equation}
\mathsf{IP}(\mathcal G_{x_2^n})\leq \mathsf{LP}(\mathcal G_{x_2^n})(1+\epsilon)^n, \qquad \forall x_2^n.\label{eqn23}
\end{equation}
This equation would then imply that
\begin{align}
\begin{split}
 \frac{1}{n}\max_{x_2^n} C(\mathcal G_{x_2^n}) &= \frac{1}{n}\max_{x_2^n} \log \mathsf{IP}(\mathcal G_{x_2^n})\\
 &\leq \frac{1}{n}\max_{x_2^n} \log \mathsf{LP}(\mathcal G_{x_2^n})+\log (1+\epsilon)\\
&= \max_{x_2}\bar{C}(\mathcal G_{x_{2}})+\log (1+\epsilon).
\end{split}
\end{align}
 It remains to show equation \eqref{eqn23}. Without loss of generality assume that $\mathcal{X}_2=\{1,2,..., r\}$ for some natural number $r$. Take an arbitrary sequence $x_2^n$. Let $n_i$ be the number of indices $i$ such that $x_{2i}=i$ for $i\in[r]$. Then $\sum_i n_i=n$ and the tuple $(n_1/n, n_2/n, ..., n_r/n)$ indicates the type of the sequence $x_2^n$. Because for any two graphs $\mathcal{G}_1$ and $\mathcal{G}_2$ we have $$\mathsf{IP}(\mathcal G_{1}\otimes\mathcal G_{2})=\mathsf{IP}(\mathcal G_{2}\otimes\mathcal G_{1}),$$
we get that 
$\mathsf{IP}(\mathcal G_{x_2^n})=\mathsf{IP}(\Motimes_{i=1}^r \mathcal G_{i}^{\otimes n_i})$.
From equations \eqref{prodID} and  \eqref{prodID2}, we have that
\begin{align}
\mathsf{IP}(\mathcal G_{x_2^n})&\leq \prod_{i=1}^r \mathsf{IP}(\mathcal G_{i}^{\otimes n_i}),\\
\mathsf{LP}(\mathcal G_{x_2^n})&=\prod_{i=1}^r \mathsf{LP}(\mathcal G_{i}^{\otimes n_i}).
\end{align}
Taking logarithms from both sides and dividing by $n$, we get
\begin{align}
\frac{1}{n}\log\mathsf{IP}(\mathcal G_{x_2^n})&\leq \sum_{i=1}^r \frac{n_i}{n}\log (\mathsf{IP}(\mathcal G_{i}^{\otimes n_i})^{1/n_i}),\label{eqn:abv-1}\\
\frac{1}{n}\log\mathsf{LP}(\mathcal G_{x_2^n})&=\sum_{i=1}^r \frac{n_i}{n}\log (\mathsf{LP}(\mathcal G_{i}^{\otimes n_i})^{1/n_i})\nonumber
\\&= \sum_{i=1}^r \frac{n_i}{n}\log (\mathsf{LP}(\mathcal G_{i})),  \label{eqn:abv}
\end{align}
Where in equation \eqref{eqn:abv} we used equation \eqref{prodID2}.	

From Lov\'asz's result in equation \eqref{LovaszThem}, we know that 
$$\lim_{n_i\to \infty} \mathsf{IP}(\mathcal G^{\otimes n_i}_{i})^{1/n_i}=\mathsf{LP}(\mathcal G_{i}).$$
Thus, given any $\delta>0$, for each $i\in[r]$, if $n_i>N_i$ for some sufficiently large threshold $N_i$, we have that 
$\log (\mathsf{IP}(\mathcal G_{i}^{\otimes n_i})^{1/n_i})$ is within $\delta$ interval of $\log (\mathsf{LP}(\mathcal G_{i}))$.
Compare equations \eqref{eqn:abv-1} and \eqref{eqn:abv}. If $n_i>N_i$, we have that the corresponding terms in the sum are within $\delta$ interval of each other. If $n_i<N_i$, by letting $n$ go to infinity, we can make $n_i/n\leq N_i/n$ as small as we want. This would complete the proof.
\end{proof}

\begin{remark}
We  show that the general lower bound given in Section \ref{sec:general-lower} is tight in this case. Consider that $X_2$ is a function of $X_1$. Then, similar to Corollary \ref{cor}, when $X_1\rightarrow F\rightarrow Y_2$ and $X_2$ is a function of $X_1$, we have 
\begin{equation}
I(F;X_1|X_2)=I(FY_2;X_1|X_2)\geq I(Y_2;X_1|X_2).
\end{equation}
Therefore, we know that
\begin{equation}
\max_{q(x_1,x_2)} \min_{p(y_2|x_1,x_2)\in \mathcal{P}} I(Y_2;X_1|X_2)
\end{equation}
is a lower bound for the coordination capacity, where the maximum is over all $q(x_1, x_2)$ such that $q(x_1, x_2)=0$ if $p(x_1, x_2)=0$. Observe that
\begin{align}
\max_{q(x_1,x_2)} \min_{p(y_2|x_1,x_2)\in \mathcal{P}} I(Y_2;X_1|X_2)&\geq \max_{x_2} \max_{q(x_1|x_2)} \min_{p(y_2|x_1,x_2)\in \mathcal{P}} I(Y_2;X_1|X_2=x_2)
\\&= \max_{x_2} \bar C(\mathcal G_{x_2}),
\end{align}
where the maximum in the second equation is over $q(x_1|x_2)$ where $q(x_1|x_2)>0$ only if $p(x_1|x_2)>0$. Thus, $ \max_{x_2} \bar C(\mathcal G_{x_2})$ is a tight lower bound to the  coordination capacity.
\end{remark}

\section{Linear Coordination} \label{sec:linear}
In this section we will study a special case of set coordination problems, where actions and inputs are constrained by linear equations. More specifically in a network with $m$ nodes, we assume that the inputs $X_i$ and output $Y_i$ are all vectors (of possibly different lengths) in a given field $\mathbb{F}$. We say that the nodes are coordinated if 
\begin{equation}
K_1X_1+K_2X_2+\cdots+K_mX_m+K_{m+1}Y_1+\cdots+K_{2m}Y_m=0,
\end{equation}
for some matrices $K_i$, $i\in[2m]$.

One motivation for this model comes from linear control systems. Suppose that nodes are controllers and inputs are disturbances to the system. Controllers should undo the disturbance by producing proper actions. 

\begin{example}Linear coordination is a generalization of linear network coding. To see this, for instance consider a network of $m$ nodes, where  node one has message of $k$ bits for node $m-1$, and node two has message of length $\ell$ bits for node $m$.
Then, we can write this as a linear coordination problem consider that
\begin{equation}
K_1=-K_{2m-1}=\begin{bmatrix}
I_{k\times k}  & 0\\
0 & 0
\end{bmatrix}_{(\ell+k)\times (\ell+k)}, \hspace{4 mm}
K_2=-K_{2m}=\begin{bmatrix}
0 & 0\\
0 & I_{\ell\times \ell}
\end{bmatrix}_{(\ell+k)\times (\ell+k)},
\end{equation}
While other $K_i$'s are zero. First $k$ bits of $X_1$ are the message for node $m-1$, and other $\ell$ bits are zero. In first $X_2$ first $k$ bits are zero and rest of them are the message for node $m$.
\end{example}

\subsection{Linear coordination capacity}
For a linearly constrained coordination problem, we can define linear or non-linear codes.  In a linear code, all encoding and decoding operations  are linear, \emph{i.e.,} the transmitted messages are constructed linearly from the inputs, and the output actions are reconstructed linearly from the messages and the inputs. On the other hand, a non-linear code allows for non-linear encoder and decoders. 

For simplicity, we only consider the two nodes network with one-way communication of Section \ref{sec:one-waycoordination}. We assume that the inputs of nodes are column vectors $X_1\in \mathbb{F}^{r_1}, X_2 \in \mathbb{F}^{r_2}$, distributed according to some joint distribution. The outputs of the two nodes are also assumed to be column vectors $Y_1\in \mathbb{F}^{s_1}, Y_2 \in \mathbb{F}^{s_2}$. The nodes are coordinated if 
\begin{equation}
K_1X_1+K_2X_2+K_3Y_1+K_4Y_2=0,\label{eqn:Ks}
\end{equation}
for some fixed matrices $K_1,K_2,K_3$ and $K_4$. These four matrices are assumed to have the same number of $c$ rows. The number of columns of $K_1, K_2, K_3$ and $K_4$ are $r_1, r_2, s_1$ and $s_2$ respectively. Similar to Section \ref{sec:3}, we further make the simplifying assumption that node two has no input $(r_2=0)$, or equivalently $K_2=0$.

The one-shot (non-linear) coordination capacity of the communication link from node one to node two, $C(K_1, K_3, K_4)$, is defined as before for the permissible action sets $$A_{x_1}=\{(y_1, y_2): K_1x_1+K_3y_1+K_4y_2=0\}, \qquad x_1\in\mathcal{X}_1.$$
On the other hand, the one-shot linear coordination capacity $C_L(K_1, K_3, K_4)$ is defined as follows:
\begin{definition}[One-shot linear coordination capacity]
 The message $M$ generated from node one is assumed to be a linear combination of coordinates of $X_1$, \emph{i.e.,} $M=SX_1$ for some matrix $S$ in $\mathbb{F}^{t\times r_1}$, where $t$ is the number of symbols that are transmitted. Actions $Y_1$ and $Y_2$ are constructed linearly according to $Y_1=AX_1$ and $Y_2=BM$. The goal is to find matrices $S$, $A$ and $B$ such that 
equation \eqref{eqn:Ks} holds, while $t$ (the number of rows of matrix $S$) is minimized. The minimum value of $t$ is called the one-shot linear coordination capacity and denoted by $C_{{L}}(K_1, K_3, K_4)$.
\end{definition}

Because linear codes are a special case of non-linear codes, we have that  $C\leq C_L\log|\mathbb{F}|$.

In an asymptotic coordination problem, we have multiple inputs and outputs $X_{1i}, Y_{1i}, Y_{2i}$ for $i\in[n]$. We say that the nodes are coordinated if
\begin{equation}
K_1X_{1i}+K_3Y_{1i}+K_4Y_{2i}=0,\qquad \forall i\in[n]\label{coordin;v}.
\end{equation}
Alternatively, if we make an overall vector of inputs $\mathbf{X}_1$ by concatenating column vectors $X_{11}, X_{12}, ..., X_{1n}$, and similarly for $\mathbf{Y}_1$ and $\mathbf{Y}_2$, we can express equation \eqref{coordin;v} as follows:
\begin{equation}
(I_{n\times n}\otimes K_1)\mathbf{X}_{1}+(I_{n\times n}\otimes K_3)\mathbf{Y}_{1}+(I_{n\times n}\otimes K_4)\mathbf{Y}_{2}=0,\qquad \forall i\in[n]\label{coordin;v1}.
\end{equation}
Therefore, we can make the following definitions:

\begin{definition}[Asymptotic linear coordination capacity]
 The asymptotic linear coordination capacity is defined as 
$$\bar C_{{L}}(K_1, K_3, K_4)=\lim_{n\rightarrow \infty}\frac 1n C_{{L}}(I_{n\times n}\otimes K_1, I_{n\times n}\otimes K_3, I_{n\times n}\otimes K_4).$$
Similarly, the  asymptotic non-linear coordination capacity can be expressed as
$$\bar C(K_1, K_3, K_4)=\lim_{n\rightarrow \infty}\frac 1n C(I_{n\times n}\otimes K_1, I_{n\times n}\otimes K_3, I_{n\times n}\otimes K_4).$$
\end{definition}

\subsection{Linear one-way coordination capacity}
In this section we study one-shot and asymptotic linear capacity for the one-way communication setup. 

\begin{theorem}\label{LTh} Let subspace $\mathscr{V}$ be the linear span of all vectors $x_1\in \mathbb{F}^{r_1}$ such that $p(x_1)>0$. Then we have 
$C_L(K_1, K_3, K_4)=\bar C_L(K_1, K_3, K_4)=f(K_1,K_3,K_4)$ where 
\begin{align}
\begin{split}
f(K_1,K_3,K_4)&=\min_{\mathscr{U}} \dim \mathscr{U}\\
\text{s.t.\hspace{3 mm}} K_1\mathscr{V}&\subseteq  \mathfrak{Im}(K_3)\oplus \mathscr{U}\\
\mathscr{U}&\subseteq  \mathfrak{Im}(K_4).
\end{split}\label{eqn:Tmh4}
\end{align}
Here the minimum is over linear subspaces $\mathscr{U}$ that satisfy the given constraints, $\oplus$ is the Minkowski sum and $\mathfrak{Im}(\cdot)$ is the image operator.
\end{theorem}
\begin{proof}[Proof of Theorem \ref{LTh}] We first prove that $C_L(K_1, K_3, K_4)=f(K_1, K_3, K_4)$ and then prove that $\bar C_L(K_1, K_3, K_4)=C_L(K_1, K_3, K_4)$.

\textbf{Proof of $C_L(K_1, K_3, K_4)=f(K_1, K_3, K_4)$:}
The coordination constraint can be written as
\begin{align}\nonumber
\begin{split}
K_1X_1+K_3Y_1+K_4Y_2&=K_1X_1+K_3AX_1+K_4BSX_1\\&=(K_1+K_3A+K_4BS)X_1=0
\end{split}
\end{align}
Thus, $(K_1+K_3A+K_4BS)x_1=0$ for all $x_1$ where $p(x_1)>0$. Therefore,
the necessary and sufficient condition for this equation to hold is that 
\begin{equation} \label{LC}
(K_1+K_3A+K_4BS)\mathscr{V}=0.
\end{equation}
We have that  $K_3A\mathscr{V}\subseteq \mathfrak{Im}(K_3)$ and $K_4BS\mathscr{V}\subseteq \mathfrak{Im}(K_4B)\triangleq\mathscr{U}$. Then since $B$ has $t$ columns we have $\dim \mathscr{U} \leq t$. Therefore, equation \eqref{LC} implies that equation \eqref{eqn:Tmh4} holds for some appropriate $\mathscr{U}$ and $C_L(K_1,K_2,K_3) \geq f(K_1,K_2,K_3)$.

To show the other direction, assume that there exist a vector space $\mathscr{U}\subseteq  \mathfrak{Im}(K_4)$ such that $K_1\mathscr{V}\subseteq  \mathfrak{Im}(K_3) \oplus \mathscr{U}$ and $\dim \mathscr{U}=t$. We will find appropriate matrices $A,B$ and $S$. Let $\{v_i\}$ be a basis for $\mathscr{V}$. Equation \eqref{LC} is true if and only if $(K_1+K_3A+K_4BS)v_i=0$. Now let $q_i=K_1v_i$. Since $K_1\mathscr{V}\subseteq  \mathfrak{Im}(K_3) \oplus \mathscr{U}$ we can find vectors $w_i\in \mathfrak{Im}(K_3)$ and $r_i\in \mathscr{U}$ such that $q_i=w_i+r_i$. 

Now we want to find matrix $A$ such that $w_i=K_3Av_i$ for all $i$. If we show $[v_1|\cdots|v_l]$ by $V$ and $[w_1|\cdots|w_l]$ by $W$. Then we should find $A$ such that $W=K_3AV$.
Since $w_i\in \mathfrak{Im}(K_3)$ we can find matrix $L$ such that $W=K_3L$. Thus we need $L=AV$. Notice that since $\{v_i\}$ are a basis for $\mathscr{V}$, matrix $V$ is full column rank. Therefore its rows span the full space. Let us denote the $i$-th row of $L$ by $l_i$, and $i$-th row of $A$ by $a_i$. We have to solve linear equation $l_i=a_iV$ where $a_i$ is a vector that we need to find. Now since span of rows of $V$ is the full space we can always find proper vectors $a_i$. Hence we can find $A$ such that $W=K_3AV$.  

Now for finding matrices $B$ and $S$ (with $t$ rows); we should have $r_i=K_4BSv_i$. Take $R=[r_1|\cdots|r_l]$. We need $R=K_4BSV$.

Since $r_i\in \mathscr{U}$, we have that $rank(R)\leq t$. 
 Since $r_i \in \mathfrak{Im}(K_4)$, we can find $L'$ such that $R=K_4L'$. Now, we want to prove that we can choose $L'$ such that $rank(L')\leq t$. Notice that each columns of $L'$ can be written as the sum of two vectors such that one of them lies in the kernel space of $K_4$ and the other one is perpendicular to this space. Therefore $L'=L'_1+L'_2$. We have that $K_4L'=K_4L'_1+K_4L'_2$.
The first part $K_4L'_1$ will vanish and $K_4L'=K_4L'_2$. Therefore, without loss of generality we can assume that $\mathfrak{Im}(L') \subseteq Ker(K_4)^{\perp}$ or $\mathfrak{Im}(L')\cap Ker(K_4)=0$. Next, it is known for arbitrary matrices $A$ and $B$ that $rank(AB)=rank(B) - \dim ( \mathfrak{Im}(B) \cap Ker(A) )$. Thus, 
\begin{equation} \label{RankL}
rank(R)=rank(L')-\dim(\mathfrak{Im}(L') \cap Ker(K_4))=rank(L').
\end{equation}
 Thus $rank(L')=rank(R)\leq t$.

Now we should have $L'=BSV$. Again notice that $V$ is full column rank and rank of $L'$ is equal to $t$. Therefore there exist $t$ rows that can produce all rows of $L'$. Choose $S$ such that $t$ rows of $SV$ can produce $L'$. Now if we denote rows of $SV$ by $s_1,\cdots,s_t$ and rows of $L'$ by $l'_1,\cdots,l'_{s_2}$, since $l'_i$ is in the span of $s_1,\cdots,s_t$, we can find real numbers $b_{ij}$ such that $l'_i=\sum_j b_{ij}s_j$ if we take $B=[b_{ij}]$, then we have $L'=B(SV)$. This completes the proof. 

\textbf{Proof of $C_L(K_1, K_3, K_4)=\bar C_L(K_1, K_3, K_4)$:} It suffices to show that 
$$\frac 12 C_{{L}}(I_{2\times 2}\otimes K_1, I_{2\times 2}\otimes K_3, I_{2\times 2}\otimes K_4)=C_L(K_1, K_3, K_4).$$
This is because the above equality can be used inductively to show that for any $n$, which is a power of two, we have
$$\frac 1n C_{{L}}(I_{n\times n}\otimes K_1, I_{n\times n}\otimes K_3, I_{n\times n}\otimes K_4)=C_L(K_1, K_3, K_4).$$
Therefore $C_L(K_1, K_3, K_4)=\bar C_L(K_1, K_3, K_4)$.

Consider that node one observe two vectors $X_{11}$ and $X_{12}$ and then sends $M=S\begin{bmatrix}X_{11}\\X_{12}\end{bmatrix}$ to the second node where $$S=\begin{bmatrix}S_1&S_2\\S_3&S_4\end{bmatrix}.$$ Node one produces $$\begin{bmatrix}Y_{11}\\Y_{12}\end{bmatrix}=\begin{bmatrix}A_1&A_2\\A_3&A_4\end{bmatrix}\begin{bmatrix}X_{11}\\X_{12}\end{bmatrix}$$ and node two produces $$\begin{bmatrix}Y_{21}\\Y_{22}\end{bmatrix}=BM=\begin{bmatrix}B_1&B_2\\B_3&B_4\end{bmatrix}\begin{bmatrix}S_1&S_2\\S_3&S_4\end{bmatrix}\begin{bmatrix}X_{11}\\X_{12}\end{bmatrix}.$$ Coordination constraint for first letter gives us
\begin{equation}
(K_1+K_2A_1+K_3(B_1S_1+B_2S_3))X_{11}+(K_2A_2+K_3(B_1S_2+B_2S_4))X_{12}=0.
\end{equation}
A similar condition holds for the second letter. Since $X_{11}$ and $X_{12}$ are independent, we get
$$(K_1+K_2A_1+K_3(B_1S_1+B_2S_3))X_{11}=(K_2A_2+K_3(B_1S_2+B_2S_4))X_{12}=0.$$
Equivalently,
$$(K_1+K_2A_1+K_3(B_1S_1+B_2S_3))\mathscr{V}=(K_2A_2+K_3(B_1S_2+B_2S_4))\mathscr{V}=0.$$
We claim that without loss of generality, we can make the following two assumptions:
\begin{itemize}
\item
 We can assume that columns of $B_i$ are perpendicular to the kernel space of $K_3$. To see this, observe that each columns of $B_i$ can be written as the sum of two vectors such that one of them lies in the kernel space of $K_3$ and the other one is perpendicular to this space. This gives us a decomposition of matrix $B_i$ as $B_i=B_{i1}+B_{i2}$ where $K_3B_{i1}=0$. Since  $K_3B_{i1}$ vanishes, only $K_3B_{i2}$ remains and we may assume that columns of $B_i$ are perpendicular to the kernel space of $K_3$.

\item we can assume that $S_i z=0$ for $i=1,2,3,4$ and for all vectors $z \in \mathscr V^{\perp}$, where $\mathscr V^{\perp}$ is the linear subspace perpendicular to $\mathscr V$.

Notice that for every matrix $S$ we can find matrix $S'$ with same dimension such that $Sv=S'v$ for all $v \in \mathscr V$; and $S'z=0$ for all  $z \in \mathscr V^{\perp}$. This is because we can define a linear function like $S':\mathbb{F}^{r_1} \to \mathbb{F}^t$ by determining effect of this function on a basis. If $\{v_1,\cdots,v_l\}$ be an orthonormal basis for $\mathscr V$ and we expand this to an orthonormal basis for $\mathbb{F}^r_1$ like $\{v_1,\cdots,v_l,z_1,\cdots,z_{r_1-l}\}$, then we can define $S'(v_i)=Sv_i$ and $S'(z_i)=0$. It is only effect of matrices on $\mathscr V$ in important for us, and we can consider that $Sz=0$ if $z \in \mathscr V^{\perp}$. 
\end{itemize}

Each column of $K_3$ can be written as the sum of two vectors such that one of them lies in the image of $K_2$ and the other one perpendicular to image of $K_2$. This gives us a decomposition of matrix $K_3$ as $K_3=K_{31}+K_{32}$, where all columns of $K_{31}$ are in image of $K_2$ and all columns of $K_{32}$ are perpendicular to this space. Now since node one knows $B$ and $S$ it can choose $A'_1$ and $A'_2$ such that
\begin{align*}
\begin{split}
K_2A'_1&=K_2A_1+K_{31}(B_1S_1+B_2S_3),\\
K_2A'_2&=K_2A_2+K_{31}(B_1S_2+B_2S_4).
\end{split}
\end{align*}
Using the above equations, we get that
\begin{align*}
\begin{split}
(K_1+K_2A'_1&+K_{31}(B_1S_1+B_2S_3))\mathscr{V}=0\\
(K_2A'_2&+K_{32}(B_1S_2+B_2S_4))\mathscr{V}=0
\end{split}
\end{align*}
The second equation implies that $(B_1S_2+B_2S_4)\mathscr{V}=0$ since there is no common vector except zero in images of $K_2$ and $K_{32}$ (all columns of $K_{32}$ were perpendicular to the image space of $K_2$). A similar argument for the second letter shows that $(B_3S_1+B_4S_3)\mathscr{V}=0$.
 
On the other hand, from the assumption that $S_1\mathscr{V}^{\perp}=S_2\mathscr{V}^{\perp}=S_3\mathscr{V}^{\perp}=S_4\mathscr{V}^{\perp}=0$, we have that
$(B_1S_2+B_2S_4)\mathscr{V}^{\perp}=(B_3S_1+B_4S_3)\mathscr{V}^{\perp}=0$. This fact, in conjunction with  $(B_1S_2+B_2S_4)\mathscr{V}=(B_3S_1+B_4S_3)\mathscr{V}=0$, implies that $B_1S_2+B_2S_4=B_3S_1+B_4S_3=0$. 

Now
\begin{equation}\nonumber
BM=\begin{bmatrix}
B_1S_1+B_2S_3 & B_1S_2+B_2S_4\\ B_3S_1+B_4S_3&B_3S_2+B_4S_4\end{bmatrix}=\begin{bmatrix}
B_1S_1+B_2S_3 & 0\\ 0&B_3S_2+B_4S_4
\end{bmatrix}
\end{equation}
So $rank(BM)=rank(B_1S_1+B_2S_3)+rank(B_3S_2+B_4S_4)$. Besides, $rank(BM)\leq 2t$. Therefore either $rank(B_1S_1+B_2S_3)\leq t$ or $rank(B_3S_2+B_4S_4)\leq t$. Thus, one can do one letter coordination with rate less than or equal to $t$. Therefore 
$$\frac 12 C_{{L}}(I_{2\times 2}\otimes K_1, I_{2\times 2}\otimes K_3, I_{2\times 2}\otimes K_4)\geq C_L(K_1, K_3, K_4).$$
 and block coding cannot help.

\end{proof}

\subsection{Linear and non-linear coordination capacities}
\begin{theorem}Assuming $\mathfrak{Im}(K_4) \subseteq \mathfrak{Im}(K_1)$, we have that $C=\bar C=C_L\log|\mathbb{F}|=\bar C_L\log|\mathbb{F}|$. In other words, linear strategies are optimal and block coding does not help decrease the non-linear coordination rate.
\end{theorem} 
\begin{proof} From Theorem~\ref{LTh}, we know that $C_L=\bar C_L$. We know that $C_L\log|\mathbb{F}|\geq C\geq \bar{C}$. The asymptotic coordination capacity $\bar{C}$ is equal to $\log \mathsf{LP}$ for the linear program of equation \eqref{eqn:LP2}, on a graph $\mathcal{G}$ constructed as follows: $\mathcal{G}$ is a bipartite graph with nodes indexed by $\mathcal{X}_1$ on one part, and by $\mathcal{Y}_2$ on the other part. Vertex $x_1$ is connected to $y_2$ if and only if there is some $y_1$ such that $K_1x_1+K_3y_1+K_4y_2=0$.

Let $\mathsf{LP}^{\dagger}$ be the dual of the linear program given in equation \eqref{eqn:LP2}. This linear program is given in \cite[page 2]{c} and consists of variables in the interval $[0,1]$. While the original LP involved a minimization and was asking for a fractional covering, the dual linear program involves maximizing a linear expression and can be understood as a fractional packing linear program. If we restrict the variables of the dual program to integers in $\{0,1\}$, we get a lower bound on  $\mathsf{LP}^{\dagger}$. We denote the answer to this integer program by $\mathsf{IP}^{\dagger}$. It is shown in \cite[page 2]{c} that $\mathsf{IP}^{\dagger}$ is the maximum number of vertices in $\mathcal{X}_1$ whose neighbor sets in the bipartite graph are disjoint.  

To sum this up, we always have the following chain of inequalities:
$$C_L\log|\mathbb{F}|\geq C\geq \bar{C}=\log\mathsf{LP}=\log\mathsf{LP}^{\dagger}\geq \log \mathsf{IP}^{\dagger}.$$

We show that \begin{align}\log \mathsf{IP}^{\dagger}\geq C_L\log|\mathbb{F}|\label{eqnlas},\end{align} which implies that all of the above inequalities are equality. Observe that both $C_L\log|\mathbb{F}|$ and $\log \mathsf{IP}^{\dagger}$ are one-shot expressions and can be computed from the graph $\mathcal{G}$ (rather than its tensor products).

Assume that two vertices $x_1$ and $x'_1$ have a common neighbor like $y_2$. Then there exist $y_1, y'_1$ such that 
\begin{equation}
K_1x_1+K_3y_1+K_4y_2=K_1x'_1+K_3y'_1+K_4y_2=0.
\end{equation}
Hence $K_1(x_1-x'_1)\in  \mathfrak{Im}(K_3)$. Furthermore, $K_1x_1, K_1x'_1\in\mathfrak{Im}(K_1)$.  Therefore, to show that $\mathsf{IP}^{\dagger}\geq N$, it suffices to find vectors $v_1, ..., v_N$ such that 
\begin{align}\label{NL}
\begin{split}
v_i&\in  \mathfrak{Im}(K_1),\\
v_i-v_k&\notin  \mathfrak{Im}(K_3),\qquad \forall i\neq k,
\end{split}
\end{align}
Let $\mathscr U$ be the vector space with minimum dimension in Theorem \eqref{LTh}.  The dimension of $\mathscr U$ is equal to $C_L$. Hence there are $|\mathbb{F}|^{C_L}$ distinct vectors in $\mathscr U$. We claim that the set of vectors in $\mathscr U$, satisfy  both conditions of equation \eqref{NL}. This would imply that $\mathsf{IP}^{\dagger}\geq |\mathbb{F}|^{C_L}$ and gives us equation \eqref{eqnlas}.  First, observe that $\mathscr U\subseteq \mathfrak{Im}(K_4)\subseteq  \mathfrak{Im}(K_1)$. Hence, the first condition of  \eqref{NL} is clearly satisfied. To show the second condition, observe that $\mathscr U$ is a vector space with minimum dimension such that $ \mathfrak{Im}(K_1)\subseteq  \mathfrak{Im}(K_3)\oplus \mathscr{U}$. We claim that this implies $\mathscr U\cap \mathfrak{Im}(K_3)=\{0\}$. Otherwise if $u\in \mathscr{U}\cap\mathfrak{Im}(K_3)$, we can expand $u$ to a basis for $\mathscr U$ like $\{u,v_1,\cdots,v_{C_L-1}\}$. Let $\mathscr U'$ be the linear span of the vectors $\langle v_1,\cdots,v_{C_L-1}\rangle$. Then $\mathscr U'$  is a subspace of $ \mathfrak{Im}(K_4)$ that satisfies $ \mathfrak{Im}(K_3)\oplus \mathscr{U} =  \mathfrak{Im}(K_3)\oplus \mathscr{U}'$. Therefore, we can decrease dimension of $\mathscr U$ which is a contradiction. This completes the proof. 
\end{proof}

\subsection{Extensions to multiple-terminal}\label{sec:extensions}
It is possible to extend the result in linear coordination to certain multi-terminal scenarios. For instance, consider a network with a broadcast channel topology depicted in Fig.~\ref{fig:BC}, where node one observes $X_1\in\mathbb{F}^{r_1}$ and sends $SX_1$ to node two and $TX_1$ to node three. All three nodes produce outputs, but nodes two and three have no inputs. Coordination constraint requires that $K_1X_1+K_4Y_1+K_5Y_2+K_6Y_3=0$. In a linear code, we have that $Y_1=BX_1$, and the outputs of nodes two and three are constructed linearly from their received messages, \emph{i.e., } $Y_2=C(SX_1)$ and $Y_3=D(TX_1)$ for some matrices $C$ and $D$. This gives us the equation
\begin{equation}
(K_1+K_4B+K_5CS+K_6DT)X_1=0.
\end{equation}

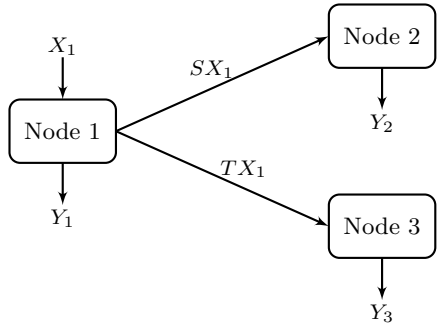
\begin{figure}[h]
\begin{center}
\begin{tikzpicture}[node distance=3.5cm,auto,>=latex', scale=1.4]
    	\draw[black, thick,rounded corners] (0,0) rectangle (1,0.6) node[font=\fontsize{9.5}{144}\selectfont,pos=.5,anchor=center] {Node 1};
	\draw[black, thick,->] (1,.3) -- (3,1.2) node[font=\fontsize{8.5}{144}\selectfont,pos=.45,anchor=south] {$SX_1$}; 
	\draw[black,thick,rounded corners](3,.9) rectangle (4,1.5) node[font=\fontsize{9.5}{144}\selectfont,pos=.5,anchor=center]{Node 2};
	\draw[black,thick,rounded corners](3,-.3) rectangle (4,-.9) node[font=\fontsize{9.5}{144}\selectfont,pos=.5,anchor=center]{Node 3};
	\draw[black, thick,->] (1,.3) -- (3,-.6) node[font=\fontsize{8.5}{144}\selectfont,pos=.6,anchor=south] {$TX_1$}; 
	\draw[black, thick,->] (0.5,1) ->(0.5,0.6) node[font=\fontsize{8.5}{144}\selectfont,pos=.2,anchor=south]{$X_1$};
	\draw[black, thick,->] (0.5,0) -> (0.5,-0.4) node[font=\fontsize{8.5}{144}\selectfont,pos=.85,anchor=north] {$Y_1$};
	\draw[black, thick,->] (3.5,.9) -> (3.5,.5) node[font=\fontsize{8.5}{144}\selectfont,pos=.85,anchor=north] {$Y_2$};
	\draw[black, thick,->] (3.5,-.9) -> (3.5,-1.3) node[font=\fontsize{8.5}{144}\selectfont,pos=.85,anchor=north] {$Y_3$};
\end{tikzpicture}
\end{center}
\caption{Linear coordination with broadcast topology.}\label{fig:BC}
\end{figure}

Thus, $(K_1+K_4B+K_5CS+K_6DT)x_1=0$ for all $x_1$ where $p(x_1)>0$. Therefore, similar to previous parts we define $\mathscr{V}$ to be the linear span of all vectors $x_1\in \mathbb{F}^{r_1}$ such that $p(x_1)>0$. Hence
\begin{equation} 
(K_1+K_4B+K_5CS+K_6DT)\mathscr{V}=0.
\end{equation}

Now with an argument similar to the one given in the proof of Theorem \ref{LTh}, we have following region for sizes of $S_{t_1\times r_1}$ and $T_{t_2\times r_1}$. The pair of $(t_1,t_2)$ is valid in one-shot case if and only if
\begin{align}\nonumber
\begin{split}
&t_1\geq \dim \mathscr{U}_1 \hspace{3 mm} t_2\geq \dim \mathscr{U}_1 \hspace{3 mm}\\
& K_1\mathscr{V} \subseteq  \mathfrak{Im} K_4 \oplus \mathscr{U}_1 \oplus \mathscr{U}_2\\
&\mathscr U_1\subseteq  \mathfrak{Im}(K_5), \hspace{3 mm} \mathscr U_2\subseteq  \mathfrak{Im}(K_6).
\end{split}
\end{align}

Next, consider a MAC channel where we assume that node one and two observe two independent vectors $X_1$ and $X_2$ respectively (see Fig.~\ref{fig:MAC}).  These two nodes send messages $SX_1$ and $TX_2$ to node three. All three nodes produce outputs linearly from their observations and their received vectors. We want to find minimum number of rows of $S$ and $T$.  
Coordination constraint is $K_1X_1+K_2X_2+K_4Y_1+K_5Y_2+K_6Y_3=0$. Assuming that $Y_1=AX_1$, $Y_2=BX_2$ and $Y_3=CSX_1+DTX_2$, we get that
\begin{equation}\nonumber
K_1X_1+K_2X_2+K_4Y_1+K_5Y_2+K_6Y_3=(K_1+K_4A+K_6CS)X_1+(K_2+K_5B+K_6DT)X_2=0.
\end{equation}

\begin{figure}[h]
\begin{center}
\begin{tikzpicture}[node distance=3.5cm,auto,>=latex', scale=1.4]
    	\draw[black, thick,rounded corners] (0,1.1) rectangle (1,1.7) node[font=\fontsize{9.5}{144}\selectfont,pos=.5,anchor=center] {Node 1};
	\draw[black, thick,->] (1,1.4) -- (3,.3) node[font=\fontsize{8.5}{144}\selectfont, pos=.5,anchor=south] {$SX_1$}; 
	\draw[black,thick,rounded corners](0,-.6) rectangle (1,-1.2) node[font=\fontsize{9.5}{144}\selectfont,pos=.5,anchor=center]{Node 2};
	\draw[black,thick,rounded corners](3,0) rectangle (4,.6) node[font=\fontsize{9.5}{144}\selectfont,pos=.5,anchor=center]{Node 3};
	\draw[black, thick,->] (1,-.8) -- (3,.3) node[font=\fontsize{8.5}{144}\selectfont ,pos=.5,anchor=south] {$TX_2$}; 
	\draw[black, thick,->] (0.5,2.1) ->(0.5,1.7) node[font=\fontsize{8.5}{144}\selectfont,pos=.2,anchor=south]{$X_1$};
	\draw[black, thick,->] (0.5,1.1) -> (0.5,.7) node[font=\fontsize{8.5}{144}\selectfont,pos=.85,anchor=north] {$Y_1$};				    
	\draw[black, thick,->] (0.5,-.2) -> (0.5,-.6) node[font=\fontsize{8.5}{144}\selectfont,pos=.2,anchor=south] {$X_2$};
	\draw[black, thick,->] (.5,-1.2) -> (.5,-1.6) node[font=\fontsize{8.5}{144}\selectfont,pos=.85,anchor=north] {$Y_2$};
	\draw[black, thick,->] (3.5,0) -> (3.5,-.4) node[font=\fontsize{8.5}{144}\selectfont,pos=.85,anchor=north] {$Y_3$};
\end{tikzpicture}
\end{center}
\caption{Linear coordination in MAC.}\label{fig:MAC}
\end{figure}
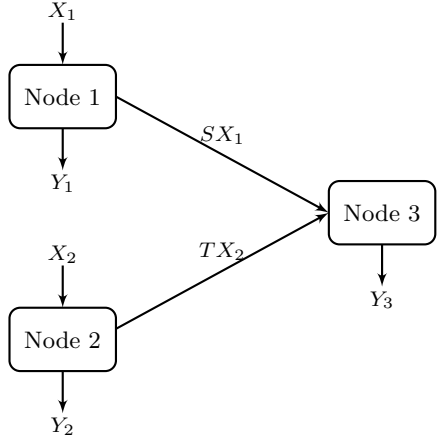

Now since $X_1$ and $X_2$ are independent both parenthesis should be zero. And these terms  are exactly same as the ones for the two nodes case. Therefore, if we denote the number of rows of $S$ and $T$ by$(t_1,t_2)$, then such a pair is valid  in one-shot case if and only if
\begin{align}\nonumber
\begin{split}
t_1\geq\dim &\mathscr{U}_1 \hspace{3 mm} t_2\geq\dim \mathscr{U}_1 \hspace{3 mm}\\
K_1\mathscr{V} &\subseteq  \mathfrak{Im}(K_4) \oplus \mathscr{U}_1, \\
K_2\mathscr{V} &\subseteq  \mathfrak{Im}(K_5) \oplus \mathscr{U}_2,\\
\mathscr U_1,\mathscr U_2&\subseteq  \mathfrak{Im} K_6.
\end{split}
\end{align}

\end{document}